\documentclass[a4paper,11pt]{article}
\pdfoutput=1 

\usepackage{jheppub} 

\usepackage[T1]{fontenc} 
\usepackage{subcaption}
\usepackage{verbatim}
\usepackage{mathrsfs}
\usepackage{physics}
\usepackage{tikz}
\usetikzlibrary{calc}
\usetikzlibrary{shapes.geometric,positioning}
\usepackage{amsthm}
\usepackage{mathtools}
\usepackage{enumerate}

\emailAdd{jcaminiti@perimeterinstitute.ca}
\emailAdd{fcapeccia@ucdavis.edu}
\emailAdd{jsorce@princeton.edu}

\definecolor{JBlue}{HTML}{1A77A6}
\definecolor{ForestGreen}{RGB}{34,139,34}

\newcommand{\ql}{``} 
\newcommand{\qr}{''$\;$}

\newcommand{\ie}{i.e.~}

\newcommand{\del}{\partial}

\renewcommand{\tilde}{\widetilde}
\renewcommand{\hat}{\widehat}

\renewcommand{\Re}{\text{Re}}

\newcommand{\A}{\mathcal{A}}

\renewcommand{\H}{\mathcal{H}}
\newcommand{\K}{\mathcal{K}}

\newtheorem{theorem}{Theorem}

\newtheorem{lemma}[theorem]{Lemma}

\theoremstyle{definition}

\newtheorem*{remark*}{Remark}

\title{Excitability of Gaussian states with VEVs}
\author[a,b]{Jacqueline Caminiti,}
\author[c]{Federico Capeccia,}
\author[d]{and Jonathan Sorce}
\affiliation[a]{Perimeter Institute for Theoretical Physics}
\affiliation[b]{Department of Physics \& Astronomy, University of Waterloo}
\affiliation[c]{University of California, Davis}
\affiliation[d]{Princeton Gravity Initiative}
\abstract{
In \cite{Caminiti:paper-2}, we gave general criteria for when one zero-mean Gaussian state can be excited out of another in a (generalized) free field theory.
Here we extend this analysis to the case of nonzero mean, i.e., to Gaussian states with vacuum expectation values (VEVs).
We prove that excitability is possible exactly when (i) the connected two-point functions satisfy criteria like those in \cite{Caminiti:paper-2}, and (ii) the difference of the VEVs  is bounded relative to the two-point functions.
As an application, we give an explicit computation showing that in anti-de Sitter spacetime, a VEV shift can be excited from the Klein-Gordon vacuum if and only if its boundary extrapolation can be excited from the vacuum of the dual conformal field theory. 
}

\begin{document}
 \maketitle

\section{Introduction}

Excitability is a central notion in quantum field theory. 
Particles are excitations of the vacuum; 
the Hilbert space associated to a reference state characterizes its possible excitations; 
and nontrivial superselection structure arises when distinct states cannot be excited out of one another using local operators. Understanding when one state of a quantum field theory can be (locally) excited out of another is therefore an important question, and one to which the algebraic approach is well-suited. 

In algebraic quantum field theory, one begins with a $\ast$-algebra $\mathcal{A}_0$ of quantum fields, living on a spacetime manifold $\mathcal{M}$. 
One then considers abstract states on $\mathcal{A}_0$ that play the role of correlation functions.
For each choice of abstract state $\omega$, the associated GNS Hilbert space $\mathcal{H}_{\omega}$  characterizes the space of local excitations around $\omega$ \cite{zbMATH03097323, Segal:irreducible}. 
Hence we say that $\omega_2$ can be excited out of $\omega_1$ if $\omega_2$ admits a representative inside $\mathcal{H}_{\omega_1}$ as a density  matrix $\rho_{\omega_2}$ satisfying
\begin{equation}
\mathrm{tr}_{\H_{\omega_1}}(\rho_{\omega_2}a)=\omega_2(a)\,, \quad a\in\A_0\,.
\end{equation}

In a prior paper \cite{Caminiti:paper-2}, we established necessary and sufficient criteria for excitability in the special case where $\A_0$ describes a (generalized) free field theory, and where $\omega_1$ and $\omega_2$ are Gaussian states with vanishing one-point functions. 
The key excitability criteria are expressible in terms of the two-point functions of $\omega_1$ and $\omega_2$, and roughly characterize the normalizability of the $\omega_2$ representative within $\mathcal{H}_{\omega_1}$.
That paper built on the work of \cite{Powers:quasi, Araki:quasi, VanDaele:quasi, Araki-Yamagami}, with improved techniques due to the canonical purification technology of \cite{Woronowicz:purification, Sorce:paper1}.

In the present paper, we generalize the results of \cite{Caminiti:paper-2} by establishing excitability criteria for Gaussian states with nonvanishing one-point functions. By abuse of terminology, we will refer to nonvanishing one-point functions as ``VEVs'' (vacuum expectation values), thinking of the state $|\omega\rangle$ as a generalized Fock vacuum within its GNS sector $\H_{\omega}.$ States with nonvanishing VEVs are physically relevant, with important examples including generic coherent states, 
as well as various states that arise in holographic CFTs.
In the pure-state case, relevant VEV excitability criteria were already stated in \cite[page 82]{Ashtekar:1987tt} and proved in \cite[proposition 2.5]{Honegger1990TheGF}; our analysis generalizes this work to arbitrary Gaussian states.

As our main result, we establish excitability criteria for general mixed Gaussian states $\omega_1$ and $\omega_2$  with VEVs.
We demonstrate that excitability requires the VEV difference of $\omega_1$ and $\omega_2$ to be bounded in an appropriate sense.
Moreover, if this VEV boundedness condition holds, we show that excitability between $\omega_1$ and $\omega_2$ is equivalent to excitability of their zero-mean counterparts, at which point \cite[theorem 1.31]{Caminiti:paper-2} can be applied.

To prove our main result, we use two key lemmas.
Lemma \ref{lem:1} establishes that simultaneously shifting a pair of Gaussian states by a single VEV does not affect whether excitability holds.
This is a simple result that follows from the construction of an intertwining unitary map.
Lemma \ref{lem:2} is more involved:
it identifies VEV boundedness as the relevant criterion for exciting a VEV out of a VEVless Gaussian state.
We address lemma \ref{lem:2} in the spirit of \cite{Caminiti:paper-2} by beginning with the pure-state case, where the desired result follows from simple calculations involving coherent states.
We then tackle the mixed-state case using canonical purifications \cite{Sorce:paper1}.
A subtlety (also encountered in \cite{Caminiti:paper-2})
is that canonical purifications of general states are not pure in the strictest sense. 
Instead, they are \textit{centrally pure} states,
which can be thought of as pure states with additional data having to do with a nontrivial center for the algebra of local observables. 
In our analysis below, we take care to address the centrally pure case, which includes classical systems (i.e. abelian algebras) as a subset.

After proving our theorem, we apply it to the example of free fields in anti-de Sitter space (AdS).
We study when a nontrivial VEV can be excited out of the AdS vacuum, and we demonstrate that our excitability relation is compatible with the bulk-boundary duality of \cite{Witten:1998qj, Gubser:1998bc, Balasubramanian:Lorentzian, Klebanov:1999tb}.

The outline of the paper is as follows.
In section \ref{ssec:summary-of-results}, we give a compact summary of our results in theorem-proof format.
In section \ref{sec:background-info}, we give a lightning review of relevant concepts from free field theory,
and we explain how to decompose our general excitability theorem into two lemmas.
In section \ref{sec:simult-VEV}, we prove lemma \ref{lem:1} regarding excitability under simultaneous VEV shifts.
In section \ref{sec:one-sided-VEV}, we prove lemma \ref{lem:2} regarding excitability under one-sided VEV shifts.
In section \ref{sec:large-N}, we discuss an application to free fields in anti-de Sitter space.

\vspace{10pt}
\noindent
\textbf{Comment on notation.} In the following, $\A_0$  denotes a (generalized) free $\ast$-algebra, 
$\omega$  denotes a Gaussian state,
and $\H_\omega$  denotes the GNS Hilbert space of $\omega$.
The symbol $\A_\omega$ denotes the von Neumann algebra associated to $\A_0$ in the GNS space $\H_\omega$.

Throughout the paper, we  work with smeared field operators, denoted $\phi[f]$, representing the formal expression
\begin{equation}
\phi[f] \equiv \int \dd{^{D} x} \sqrt{|g|} \phi(x) f(x),
\end{equation}
with $\phi(x)$ the (generalized) free field. Correlation functions are written as multidistributions $\omega(\phi(x_1)\dots\phi(x_n))$, with the understanding that the actual correlators of field operators are obtained by smearing against test functions.
Similarly, a VEV shift $v$ can be specified by a function (or distribution) $v_s(x)$, so that one has
\begin{equation}
    v[f]\equiv \int \dd{^{D} x} \sqrt{|g|} v_s(x) f(x).
    \label{eq:v-functional}
\end{equation}
The subscript \ql$s$'' indicates that, in the case of a Klein-Gordon theory, $v_s$ is required to be a solution of the equations of motion.

\subsection{Summary of results}
\label{ssec:summary-of-results}

\begin{remark*}
By a ``smearing function'' $f$, we always mean a smooth, compactly supported, complex-valued function on spacetime.
A Gaussian state $\omega$ induces a natural inner product on the space of smearing functions, denoted by $\mu$ and defined explicitly in section \ref{ssec:free-fields}.
A VEV will always be a linear functional on the space of smearing functions, written $v[f]$, and when we say that a VEV is ``$\mu$-bounded'' we mean there is a constant $C$ satisfying, for all $f$,
\begin{equation}
    |v[f]|^2 \leq C \langle f | f\rangle_{\mu}.
\end{equation}
Equivalently, $v[f]$ is $\mu$-bounded if the convergence $\langle f_n | f_n\rangle_{\mu} \to 0$ implies $v[f_n] \to 0.$
\end{remark*}

\begin{remark*}
Borrowing notation from \cite{Caminiti:paper-2}, we write $\omega_2 \prec \omega_1$ when $\omega_2$ can be excited from $\omega_1.$ A mathematically precise definition of excitability can be found in section 2 of that paper.
\end{remark*}

\setcounter{theorem}{0}
\begin{theorem}
\label{thm:1}
    Given two Gaussian states $\omega_1$ and $\omega_2$ on a (generalized) free field $\ast$-algebra $\A_0$,
    $\omega_2$ can be excited out of  $\omega_1$ if and only if
    \begin{enumerate}
        \item[(i)] the same is true for their zero-mean counterparts, namely for the zero-mean Gaussian states whose two-point functions are given by the connected two-point functions of $\omega_1$ and $\omega_2$; and
        \item[(ii)] the difference of one-point functions, $\omega_2(\phi[f])-\omega_1(\phi[f])$, is $\mu_1$-bounded.
    \end{enumerate}
\end{theorem}
\noindent

\begin{proof}
The necessity of condition (ii) is proved in section \ref{ssec:nec-crit}.
The remainder of the theorem follows from combining lemmas \ref{lem:1} and \ref{lem:2} below, as explained in section \ref{ssec:proof-strategy}.
\end{proof}

\setcounter{theorem}{0}
\begin{lemma}
\label{lem:1}
    Excitability is preserved under any simultaneous VEV shift; that is,
\begin{equation}
    \omega_2 \prec \omega_1 \iff \omega_2^v \prec \omega_1^v\,,
\end{equation}
where the superscript $v$ indicates the VEV-shifting operation introduced in section \ref{ssec:proof-strategy}.
\end{lemma} 

\begin{proof}
See section \ref{sec:simult-VEV}.
\end{proof}

\begin{lemma}
\label{lem:2}
    A VEV $v$ can be excited out of a zero-mean Gaussian state $\omega$ if and only if $v$ is $\mu$-bounded; that is, 
\begin{equation}
    v[f] \text{ is $\mu$-bounded} \quad\Longleftrightarrow\quad \omega^v \prec \omega .
\end{equation}
\end{lemma}

\begin{proof}
See section \ref{sec:one-sided-VEV}.
\end{proof}

\section{Setting up the toolkit}
\label{sec:background-info}

In this section, we briefly review the concepts in (generalized) free field theory needed below. We then prove that for Gaussian states, excitability requires the VEV difference to be appropriately bounded, and we explain how this will be used to reduce the problem of excitability for general Gaussian states to that for zero-mean Gaussian states.
We refer the reader to \cite[appendix A]{Caminiti:paper-2} or \cite{Fewster:2019ixc} for a more detailed introduction to the background material reviewed here.

\subsection{(Generalized) free field theory}
\label{ssec:free-fields}

A generalized free theory of a real scalar field $\phi$ is defined by an abstract $\ast$-algebra $\A_0$ generated by smeared fields $\phi[f]$. These fields satisfy linearity relations, the reality condition $\phi[f]^*=\phi[f^*]$, and the c-number commutation relation
\begin{equation}\label{eq:CCR}
    [\phi[f],\phi[g]]=-i\Omega[f,g].
\end{equation}
Different generalized free theories correspond to different choices of  the antisymmetric functional $\Omega$.
The theory of a Klein-Gordon scalar arises as a special case; there, $\Omega$ is chosen to be the Poisson bracket of the classical smeared field observables, and the equation of motion is imposed as a further algebraic relation.

The algebraic relations obeyed by smeared fields do not specify a canonical Hilbert space.
Instead, given a set of correlation functions $\omega$, a natural representation of $\A_0$ is provided by the GNS construction. The resulting Hilbert space $\H_\omega$ is the set of states accessible from $\omega$ by local excitations.
Within $\H_\omega$, one may then complete $\A_0$ into a  von Neumann algebra $\A_\omega$.\footnote{See \cite[appendix A]{Sorce:paper1} for a general explanation of how to complete a $*$-algebra into a von Neumann algebra within one of its GNS representations.
} 

Among the possible choices of state on $\A_0$, the simplest are Gaussian states. 
One starts with a set of correlators $\omega$ satisfying, as a formal power series in $\phi[f]$, the Gaussian property
\begin{equation}\label{eq:Gaussian-state}
\omega (e^{i\phi[f]}) =e^{-\frac{1}{2}\omega(\tilde\phi[f]^2)}e^{i\omega(\phi[f])},
\end{equation}
where $\tilde\phi[f]=\phi[f]-\omega(\phi[f])$ denotes the mean-subtracted field.
Expanding both sides in powers of $\phi[f]$ and matching order by order shows that all correlators of the form $\omega(\phi[f]^n)$ are determined by $\omega(\phi[f])$ and $\omega(\tilde\phi[f]^2)$.
By the multilinear polarization identities, together with the commutation relations \eqref{eq:CCR}, the same holds for correlators involving smearing functions which are not all equal.
Explicitly, correlators of mean-subtracted fields factorize into two-point functions, so that \eqref{eq:Gaussian-state} compactly encodes the familiar combinatorics of Wick contractions.

For equation \eqref{eq:Gaussian-state} to define the correlation functions of a quantum state, one needs to verify the positivity condition $\omega(a^*a)\ge 0$ for any $a\in \mathcal{A}_0$.
The c-number commutation relation \eqref{eq:CCR} guarantees that positivity is satisfied   
whenever it holds at the level of two-point functions of mean-subtracted fields, provided that one has $\omega(\phi[f])^*=\omega(\phi[f^*])$. 
When this is the case, we say that $\omega$ defines a Gaussian state on $\A_0$.

Consistency with equation \eqref{eq:CCR} fully fixes the antisymmetric part of the two-point function of $\omega$. Consequently, the two-point function is completely specified by its symmetric part.
This symmetric part defines a positive semidefinite inner product $\mu$ on the space of smearing functions, namely
\begin{equation}\label{eq:mu-def}
    \braket{f}{g}_\mu = \frac{1}{2}\left( \omega(\tilde\phi[f^*]\tilde\phi[g]) + \omega(\tilde\phi[g]\tilde\phi[f^*])\right).
\end{equation}
Quotienting by $\mu$-null smearing functions and taking limits leads to a Hilbert space $\K_\mu$, which encodes all the information about $\omega$ that is relevant for excitability.
Moreover, on $\K_\mu$, one can introduce an operator $R$ implementing the commutator $\Omega$ via the formula
\begin{equation}
\label{eq:R-operator-definition}
    \matrixel{f}{R}{g}_\mu=\frac{1}{2}\Omega[f^*, g].
\end{equation}
The operator $R$ is bounded, maps real smearing functions to real smearing functions, is anti-hermitian, and has spectrum lying between $-i$ and $i$.
Furthermore, the spectrum of $R$ diagnoses the purity of $\omega$: it is equal to $\{-i, i\}$ when $\omega$ is pure (\ie when its GNS representation introduces no purifying degrees of freedom, $\A_\omega'=\mathbb C$), and is equal to $\{-i, 0, i\}$ when $\omega$ is centrally pure (\ie when the only purifying degrees of freedom are limits of the original ones, $\A_\omega'\subseteq\A_\omega$).

Centrally pure states are simpler to study than general mixed states, and their importance stems from the fact that any state $\omega$ (even a non-Gaussian state in an interacting theory) can be treated as a subsystem of a centrally pure state $\hat \omega$, called the canonical purification  \cite{Sorce:paper1}. Physically, $\hat{\omega}$ corresponds to two copies of $\omega$ sewn together after a one-sided time-reversal, i.e., $\hat{\omega}$ is a state on the $*$-algebra $\hat{\A}_0 = \A_0 \otimes \A_0^{\text{op}}$. 

For generalized free fields, the $\ast$-algebra $\hat{\A}_0$ is naturally identified with a free $\ast$-algebra modeled on
the direct sum of two copies of the original space of smearing functions.
Under this identification, $\hat \omega$ is Gaussian whenever $\omega$ is Gaussian,\footnote{Explicitly, correlators in the state $\hat \omega$ are defined by the formal power series $\hat \omega(e^{i\phi[f\oplus g]})=\expval{e^{i\phi[f]}J_\omega e^{-i\phi[g]}}{\omega}$, where $J_\omega$ is the modular conjugation of $\omega$ with respect to $\A_\omega$ \cite[equation (5.2)]{Caminiti:paper-2}. Rewriting this in terms of mean-subtracted fields, we get $\expval{e^{i\tilde\phi[f]}J_\omega e^{-i\tilde\phi[g]}}{\omega} e^{i(\omega(\phi[f])+\omega(\phi[g]))}$. Using equation (5.16) in \cite{Caminiti:paper-2} --- where the tildes are omitted only because $\omega$ is zero-mean there --- we see that this has the form \eqref{eq:Gaussian-state} appropriate for a Gaussian state, with the one- and two-point functions specified in equations \eqref{eq:omega-hat-VEV} and \eqref{eq:mu-hat}.} with one-point function given by 
\begin{equation}\label{eq:omega-hat-VEV}
    \hat \omega(\phi[f_1\oplus f_2])=\omega(\phi[f_1])+\omega(\phi[f_2])
\end{equation}
and two-point function inducing the $\hat{\mu}$ inner product
\begin{equation}\label{eq:mu-hat}
    \braket{g_1\oplus g_2}{f_1\oplus f_2}_{\hat \mu}= \braket{g_1}{f_1}_\mu + \braket{g_2}{f_2}_\mu + \matrixel{g_1}{\sqrt{1+R^2}}{f_2}_\mu+\matrixel{g_2}{\sqrt{1+R^2}}{f_1}_\mu.
\end{equation}
This construction provides a key ingredient in our proofs thanks to the fact that excitability is preserved under canonical purification --- see \cite[theorem 1.11]{Caminiti:paper-2}.

Having understood the basic ingredients of (generalized) free theory, we now discuss coherent states, which are a useful class of Gaussian states with nonvanishing VEVs.
When discussing coherent states, it is useful to note that
in a (generalized) free field theory, the von Neumann algebra $\mathcal{A}_{\omega}$ is generated by the Weyl operators $e^{i \phi[f]}$ for real smearing functions $f$, which define unitaries on $\mathcal{H}_{\omega}$.
Moreover, because $\A_{\omega}$ is constructed by taking limits, it contains Weyl operators associated with any real element $|\psi\rangle_{\mu}$ of the completed phase space $\K_{\mu}$, and not just with smearing functions. 

Given a Gaussian state $\omega$ with vanishing VEV, one can define a coherent state\footnote{Using the split into creation and annihilation operators $\phi[f] = a[f]+a[f^*]^{\dag}$ (see for example appendix A.4 and C.2 of \cite{Caminiti:paper-2}), one has 
$a[f] e^{i \phi[g]}\ket{\omega}=i \omega(\phi[f]\phi[g]) e^{i \phi[g]}\ket{\omega}$.
See appendix \ref{app:abelian-excitability-combinatorics}.
This shows that equation \eqref{eq:Weyl-for-coherent} is compatible with the standard notion that coherent states are eigenstates of annihilation operators.} for each real element $\psi$ of $\mathcal{K}_{\mu}$ via the formula
\begin{equation}
  \ket{\omega_\psi}\equiv e^{i \phi[\psi]}\ket{\omega}.
  \label{eq:Weyl-for-coherent}
\end{equation}
For any real test function $f$, we have
\begin{equation}
\begin{aligned}
\langle \omega_\psi| e^{i \phi[f]}|
\omega_\psi\rangle
 &=
\langle \omega|e^{-i \phi[\psi]} e^{i \phi[f]} e^{i \phi[\psi]}|\omega\rangle\\
&=
\langle \omega| e^{i \phi[f]}|\omega\rangle e^{-i\Omega[\psi, f]}
\\
&=
\langle \omega| e^{i \phi[f]}|\omega\rangle
e^{-2i\langle \psi|R|f\rangle_{\mu}},
\end{aligned}
\end{equation}
where in the second equality we repeatedly applied the Baker-Campbell-Hausdorff formula, and in the final equality we used equations \eqref{eq:Gaussian-state} and \eqref{eq:R-operator-definition}.
Comparing again to equation \eqref{eq:Gaussian-state}, we conclude that the coherent state $\omega_\psi$ is a Gaussian state with VEV given by 
\begin{equation}
    \omega_\psi(\phi[f])
    =
    -2 \langle \psi|R|f\rangle_{\mu}\,,
    \label{eq:coherent-VEV-result}
\end{equation}
which in particular is fully specified by the choice of real vector $\ket{\psi}\in \mathcal{K}_{\mu}$.\footnote{
To see that the one-point function \eqref{eq:coherent-VEV-result} obeys the appropriate reality condition, we use equation \eqref{eq:R-operator-definition} to write
$
\omega_\psi(\phi[f])^*
    =
    -\Omega[\psi,f]^*
    =
     -\Omega[\psi,f^*]
     =
     \omega_\psi(\phi[f^*]),
$
where the second equality holds by equation \eqref{eq:CCR} and by the reality of $\ket{\psi}_{\mu}$.
}

So given an initial zero-mean Gaussian state $\omega$, 
and a real vector $|\psi\rangle_{\mu}$ in $\K_{\mu}$,
one can canonically construct a new Gaussian state $\omega_{\psi}$ with VEV determined by $|\psi\rangle_{\mu}$.
Note that all the VEVs constructed in this way are bounded with respect to $\mu$,
which follows immediately from the Cauchy-Schwarz inequality for the $\mu$ inner product.
In fact, 
we will see later that whenever a state has a $\mu$-bounded VEV,
it can always be realized, in a particular sense, as coming from a coherent state.
This is a key ingredient in the proofs to follow about excitability of Gaussian states with VEVs.

\subsection{A necessary criterion for excitability}
\label{ssec:nec-crit}

As a first step towards proving our excitability theorem \ref{thm:1}, 
we show that whenever it is possible to excite $\omega_2$ out of $\omega_1$, then the difference of their VEVs,
\begin{equation}
    v_{12}[f]=\omega_2(\phi[f])-\omega_1(\phi[f])\,,
\end{equation}
must be bounded as a functional on $\K_{\mu_1}$.

Excitability means we have a density matrix representative $\rho_{\omega_2}$ for  $\omega_2$ in the GNS Hilbert space of $\omega_1$. 
Then, for any real $f$,
\begin{equation}\label{eq:assume-excit}
    \tr_{\H_{\omega_1}}( \rho_{\omega_2}e^{i\phi[f]}) = \omega_2(e^{i\phi[f]})=e^{-\braket{f}_{\mu_2}/2}e^{i\omega_2(\phi[f])},
\end{equation}
where in the second equality we used equation \eqref{eq:Gaussian-state} together with \eqref{eq:mu-def}.
Applying this identity to a sequence of real smearing functions $f_n$ converging to zero in the $\mu_1$ topology, and isolating the VEV difference, gives
\begin{equation}
    e^{iv_{12}[f_n]}=\tr_{\H_{\omega_1}}\big( \rho_{\omega_2}e^{i\tilde\phi^{(1)}[f_n]}\big)e^{\braket{f_n}_{\mu_2}/2}\,.
    \label{eq:converging-to-1}
\end{equation}
Here we have defined the $\omega_1$-VEV shifted field $\tilde\phi^{(1)}[f]=\phi[f]-\omega_1(\phi[f])$.
Both factors on the right-hand side of equation \eqref{eq:converging-to-1} converge to one.
The trace factor does so because the map $f\mapsto e^{i\tilde\phi^{(1)}[f]}$ is continuous when one uses the $\mu_1$ topology on the domain and the ultraweak topology on the codomain (see \cite[equation (A.86)]{Caminiti:paper-2});
by taking absolute values it then follows that the exponential factor involving $\braket{f_n}_{\mu_2}$ also converges to one. 

Thus, whenever we have $\omega_2 \prec \omega_1$, \textit{any} real sequence $f_n$ that converges to zero in the $\mu_1$ inner product satisfies
\begin{equation}\label{eq:exp-i-v}
    e^{i v_{12}[f_n]}\to 1.
\end{equation}
We will now show that $v_{12}[f_n]$ must converge to zero for every such sequence, and therefore that $v_{12}$ is $\mu_1$-bounded.
We argue by contradiction: assuming that $v_{12}$ is unbounded, we choose a real sequence $g_n$ with $\braket{g_n}_{\mu_1}=1$ and $v_{12}[g_n]\to\infty$.
We then fix a real number $x$ which is not an integer multiple of $2 \pi,$ and define $f_n =\frac{x}{v_{12}[g_n]}g_n$.\footnote{Technically, we may need to pass to a subsequence to guarantee $v_{12}[g_n]\neq 0$.}
Then $f_n$ is real and converges to zero in the $\mu_1$ inner product, but we have $v_{12}[f_n]=x$ for all $n$, and hence $e^{iv_{12}[f_n]}=e^{ix}\neq 1$. This contradicts equation \eqref{eq:exp-i-v}, so we conclude that the VEV difference $v_{12}$ must be $\mu_1$-bounded.

\subsection{Proof strategy for excitability criteria}
\label{ssec:proof-strategy}

In the previous subsection, we showed that the excitability relation $\omega_2\prec\omega_1$ imposes a nontrivial constraint on the one-point functions of the states. Our goal is now to show that this constraint, together with the zero-mean excitability criteria of \cite{Caminiti:paper-2} applied to the VEV-subtracted two-point functions, give necessary and sufficient conditions for excitability. The key step is to leverage the necessary condition found in the previous section to reduce the excitability problem for Gaussian states with nonzero mean to the corresponding zero-mean problem.

To do so, we first introduce notation for shifting the one-point function of a Gaussian state while leaving its connected two-point function unchanged. Let $\omega$ be a Gaussian state. In equation \eqref{eq:Gaussian-state}, the one-point function enters only through the phase factor. Hence, given a linear functional $v$ satisfying the reality condition $v[f]^*=v[f^*]$, we define a VEV-shifted Gaussian state $\omega^v$ by
\begin{equation}\label{eq:omega-v-def}
    \omega^v(e^{i\phi[f]})=\omega(e^{i\phi[f]})e^{-iv[f]}.
\end{equation}
From this, correlation functions of fields $\phi[f]$ in the state $\omega^v$ coincide with those of the shifted fields $(\phi[f]-v[f])$  in the state $\omega$.\footnote{Since connected $n$-point functions for $n\ge 2$ are insensitive to c-number shifts, $\omega$ and $\omega^v$ share the same connected $n$-point functions for all $n\ge 2$. 
For $n\ge3$, this provides a clean restatement of the fact that $\omega^v$ is Gaussian if and only if $\omega$ is, since Gaussian states can equivalently be characterized by having vanishing connected $n$-point functions for $n\geq 3$.}
This set of correlators inherits positivity from $\omega$, and hence defines a Gaussian state with one-point function
\begin{equation}
    \omega^v(\phi[f])=\omega(\phi[f])-v[f]\,.
\end{equation}
Note that choosing $v[f]=\omega(\phi[f])$ subtracts away the full one-point function, yielding a zero-mean version of $\omega$, whose two-point function is given by the \textit{connected} two-point function of $\omega$.\footnote{
Explicitly, in the case where $v[f]=\omega(\phi[f])$, one has
\begin{equation}
\begin{aligned}
    \omega^v(\phi[f]\phi[g])
    =
    \omega\left((\phi[f]-\omega(\phi[f]))(\phi[g]-\omega(\phi[g]))\right)
    =
    \omega(\phi[f]\phi[g])-\omega(\phi[f])\omega(\phi[g])\,.
\end{aligned}
\label{eq:omega-v-as-connected-omega}
\end{equation}}
In particular, this means that $\omega^v$ and $\omega$ induce the same $\mu$ inner product on the space of smearing functions (equation \eqref{eq:mu-def}).

Given two Gaussian states $\omega_1, \omega_2$, with respective means $v_1, v_2$, our main claim (theorem \ref{thm:1}) is that excitability reduces to its zero-mean counterpart --- which was solved in \cite{Caminiti:paper-2} --- up to a boundedness condition on the difference of one-point functions.
That is,
\begin{equation}\label{eq:main-claim}
\quad \quad\omega_2 \prec \omega_1 \quad \Longleftrightarrow\quad 
\begin{aligned}
&\omega_2^{v_2} \prec \omega_1^{v_1}, \;\text{and}\\
&
\omega_2(\phi[f]) - \omega_1(\phi[f])
 \text{ is }\mu_1\text{-bounded}.
\end{aligned}
\end{equation}
To prove this, we first note that 
the necessity of the VEV boundedness criterion was proven in section \ref{ssec:nec-crit}.
Accordingly, it suffices to prove
\begin{equation}\label{eq:main-claim-reduced}
\quad \quad\omega_2 \prec \omega_1 \iff  
\omega_2^{v_2} \prec \omega_1^{v_1}, \quad \text{ if $\omega_2(\phi[f]) - \omega_1(\phi[f])$ is $\mu_1$-bounded.}
\end{equation}

We then break the proof of \eqref{eq:main-claim-reduced} into two parts, established in sections \ref{sec:simult-VEV} and \ref{sec:one-sided-VEV} respectively.
First, we prove (as lemma \ref{lem:1}) that excitability is preserved under any simultaneous VEV shift:
\begin{equation}
    \omega_2 \prec \omega_1 \iff \omega_2^v \prec \omega_1^v.
\end{equation}
Second, we prove that any $\mu$-bounded VEV can be excited out of a zero-mean Gaussian state:
\begin{equation}
    v[f] \text{ is } \mu\text{-bounded} \implies  
     \omega^v \prec \omega,
\end{equation}
thereby completing the proof of lemma \ref{lem:2}.
We get equation \eqref{eq:main-claim-reduced} by combining these two facts as follows.

To show the $\implies$ direction of equation \eqref{eq:main-claim-reduced}, assume $\omega_2\prec\omega_1$. By lemma \ref{lem:1}, applied with $v=v_2$, we have $\omega_2^{v_2}\prec \omega_1^{v_2}$.
Moreover, since $v_1-v_2$ is $\mu_1$-bounded by assumption, lemma \ref{lem:2} implies
$\omega_1^{v_2}\prec \omega_1^{v_1}$.
Thus
\begin{equation}
\omega_2^{v_2}\prec \omega_1^{v_2}\prec \omega_1^{v_1},
\end{equation}
and by transitivity of excitability we conclude $\omega_2^{v_2}\prec\omega_1^{v_1}$.

Conversely, assume $\omega_2^{v_2}\prec\omega_1^{v_1}$. By lemma \ref{lem:1}, undoing the common shift by $v_2$, we obtain $\omega_2\prec \omega_1^{v_1-v_2}$.
Since $v_1-v_2$ is $\mu_1$-bounded, lemma \ref{lem:2} then implies $\omega_1^{v_1-v_2}\prec \omega_1$.
Therefore
\begin{equation}
\omega_2\prec \omega_1^{v_1-v_2}\prec \omega_1,
\end{equation}
and transitivity gives $\omega_2\prec\omega_1$.

Lemmas \ref{lem:1} and \ref{lem:2} will be proved in the following sections.
For now, we  conclude this section with a remark.
In \cite{Caminiti:paper-2}, it was found that excitability of zero-mean Gaussian states in free theories always goes both ways; that is, if $\omega_2^{v_2}$ is excitable from $\omega_1^{v_1}$, then the inverse statement holds.\footnote{This is not true in the general setting; see \cite{Caminiti:paper-2} for examples of one-way excitability in QFT.
}
Assuming our main claim in \eqref{eq:main-claim}, we may now extend this result to general Gaussian states.
For if we have $\omega_2 \prec \omega_1$, then equation \eqref{eq:main-claim} implies $\omega_2^{v_2} \prec \omega_1^{v_1}$, and the result in \cite{Caminiti:paper-2} gives $ \omega_1^{v_1}\prec \omega_2^{v_2}$.
As discussed in \cite{Caminiti:paper-2}, mutual excitability implies that $\mu_1$ and $\mu_2$ define equivalent topologies on the space of smearing functions. 
Hence, the VEV boundedness condition required for $\omega_2 \prec \omega_1$ in equation \eqref{eq:main-claim} can be equivalently rewritten in terms of the $\mu_2$ inner product.
Thus, if we have $\omega_2 \prec \omega_1$, then we have all the criteria specified in equation \eqref{eq:main-claim} for the converse relation $\omega_1 \prec \omega_2$ to hold.

\section{Excitability under simultaneous VEV shifts}
\label{sec:simult-VEV}

Let $\omega_1$ and $\omega_2$ be two Gaussian states on the generalized free field $\ast$-algebra $\A_0$.
In this section we prove lemma \ref{lem:1}, namely that excitability is preserved under an arbitrary VEV shift $v$, provided the shift is applied \ql simultaneously'' to both states:
\begin{equation}
    \omega_2 \prec \omega_1 \iff \omega_2^v \prec \omega_1^v.
\end{equation}
It suffices to prove one implication, since the converse follows by replacing $v$ with $-v$.
We therefore assume that $\omega_2$ can be excited out of $\omega_1$.
This means --- adapting the algebraic characterization of excitability given in \cite[proposition 1.2]{Caminiti:paper-2} to the free-field setting (see also  \cite[remark 1.3]{Caminiti:paper-2}) --- that we have ultraweak continuity of the linear map $\alpha$ that takes a Weyl operator $e^{i\phi[f]}$ on $\H_{\omega_1}$ to the \ql same'' Weyl operator $e^{i\phi[f]}$ on $\H_{\omega_2}$.

We want to show that the same continuity property holds for the analogous \ql identity map'' $\alpha^v$ that sends Weyl operators on $\H_{\omega_1^v}$ to Weyl operators on $\H_{\omega_2^v}$.
It will be useful to notice that by linearity, the original map  $\alpha$ acts as the identity also on $v$-shifted Weyl operators; that is, it maps $e^{i \tilde{\phi}^v[f]}$ on $\H_{\omega_1}$ to $e^{i \tilde{\phi}^v[f]}$ on $\H_{\omega_2}$, 
where as in the preceding section we have introduced $\tilde\phi^v[f]=\phi[f]-v[f]$.

This rewriting makes apparent the existence of ultraweakly continuous intertwiners $\mathbf{U}_j$ rendering the following diagram commutative
\begin{center}
\begin{tikzpicture}[>=stealth, line width=.7pt]

\node (omega1) at (1,3.3) {$e^{i\tilde{\phi}^v[f]}\circlearrowleft \H_{\omega_1}$};
\node (alpha) at (3.1,3.55) {$\alpha$};
\node (omega2)  at (5.2,3.3) {$e^{i\tilde{\phi}^v[f]}\circlearrowleft \H_{\omega_2}$};
\node (alphav) at (3.1,1.1) {$\alpha^v$};

\node (omega1v) at (1,0.8) {$e^{i\phi[f]}\circlearrowleft \H_{\omega_1^v}$};

\node (omega2v)  at (5.3,0.8) {$e^{i\phi[f]}\circlearrowleft \H_{\omega_2^v}$};

\draw[->] (omega2.south) -- (omega2.south |- omega2v.north);
\node at (5.5,2.1) {$\,\,\mathbf{U}_2$};
\draw[->] (omega1.south) -- (omega1.south |- omega1v.north);
\node at (0.65,2.1) {$\mathbf{U}_1$};

\draw[->] (omega1.east) -- (omega2.west |- omega2.west);
\draw[->] (omega1v.east) -- (omega2v.west |- omega2v.west);

\end{tikzpicture}
\end{center}
so that $\alpha^v$ is related to $\alpha$ by
\begin{equation}\label{eq:alphav-alpha}
    \alpha^v= \mathbf{U}_2 \circ \alpha \circ \mathbf{U}_1^{-1}.
\end{equation}
Concretely, we construct $\mathbf U_j$ via conjugation by the unitary $U_j:\H_{\omega_j}\to\H_{\omega_j^v}$ defined on the set of coherent states via
\begin{equation}
    U_je^{i\tilde\phi^v[f]}\ket{\omega_j} = e^{i\phi[f]}\ket*{\omega_j^v}.
\end{equation}
To see that $U_j$ extends to an isometry, it is enough to check that it preserves inner products on this total set.
Using the Baker--Campbell--Hausdorff formula together with \eqref{eq:omega-v-def}, one verifies that indeed we have
\begin{equation}
\bra*{e^{i\tilde\phi^v[g]}\omega_j}\ket{e^{i\tilde\phi^v[f]}\omega_j}
    =
\bra*{e^{i\phi[g]}\omega_j^v} 
\ket*{e^{i\phi[f]}\omega_j^v}.
\end{equation}
Thus $U_j$ extends to an isometry from $\H_{\omega_j}$ into $\H_{\omega_j^v}$. Its range contains the dense set of vectors spanned by the coherent states $e^{i\phi[f]}\ket*{\omega_j^v}$; consequently, $U_j$ is unitary. 

To see that $\mathbf{U}_j$ acts on Weyl operators as desired, we write
\begin{equation}
    \begin{aligned}
        U_j e^{i\tilde\phi^v[f]} U_j^{-1} e^{i\phi[g]}\ket*{\omega_j^v} &= U_j e^{i\tilde\phi^v[f]}e^{i\tilde\phi^v[g]}\ket{\omega_j}\\
        &=e^{\frac{i}{2}\Omega[f,g]}U_je^{i\tilde\phi^v[f+g]}\ket{\omega_j}\\
        &=e^{\frac{i}{2}\Omega[f,g]}e^{i\phi[f+g]}\ket*{\omega_j^v}\\
        &= e^{i\phi[f]} e^{i\phi[g]}\ket*{\omega_j^v}.
    \end{aligned}
\end{equation}
So indeed $\mathbf{U}_j$ sends any shifted Weyl operator $e^{i \tilde\phi^v[f]}$ on $\H_{\omega_j}$  to the corresponding Weyl operator $e^{i \phi[f]}$ on $\H_{\omega_j^v}$ --- from this, equation \eqref{eq:alphav-alpha} immediately follows.

Conjugation by a unitary is ultraweakly continuous, therefore so are the maps $\mathbf{U}_j$. 
Equation \eqref{eq:alphav-alpha} then expresses $\alpha^v$ as a composition of ultraweakly continuous maps, so that $\alpha^v$ itself must be ultraweakly continuous, and we have $\omega_2^v \prec \omega_1^v$.

\section{Excitability under one-sided VEV shifts}
\label{sec:one-sided-VEV}

In this section, we investigate when it is possible to excite a VEV out of a zero-mean Gaussian state $\omega$.
From section \ref{ssec:nec-crit}, we already have that for a VEV to be excited out of a general Gaussian state,
it must be $\mu$-bounded.
Our goal in this section is to prove that, in fact, \textit{any} bounded VEV can be excited out of a zero-mean Gaussian state:
\begin{equation}
  v[f] \text{ is bounded with respect to } \mu\implies  \omega^v \prec \omega \,,
\end{equation}
thereby proving lemma \ref{lem:2}.
We establish this relation using canonical purification techniques along the lines of \cite{Caminiti:paper-2}.

Explicitly, we begin with two simple cases --- pure Gaussian states  on a general free $*$-algebra $\A_0$, and general Gaussian states on an abelian $\ast$-algebra $\A_0$.
We then combine these cases to prove the result for centrally pure $\omega$.
Leveraging the characterization of excitability via canonical purifications from \cite{Caminiti:paper-2}, we then prove the general case of the lemma.

\subsection{Pure state excitability}

From equation \eqref{eq:coherent-VEV-result}, we have that with $\ket{\psi}$ a general real vector in $\K_\mu$, the  coherent state
\begin{equation}
    \ket{\omega_\psi}\equiv e^{i \phi[\psi]}\ket{\omega}
\end{equation}
in $\mathcal{H}_{\omega}$ has VEV given by
\begin{equation}
    \omega_{\psi}(\phi[f]) = -2 \langle \psi |R|f\rangle_{\mu}=2\braket{R\psi}{f}_\mu,
\end{equation}
where in the second equality we used the anti-hermiticity of $R$.
In the present context, we are interested in the state $\omega^v$, which has VEV given by a bounded functional  $-v[f]$.
Crucially, by the Riesz representation theorem, boundedness of $v[f]$ with respect to $\mu$ means that the VEV shift lives as a (real)\footnote{ To check that $\ket{v}_{\mu}$ is real, write it in terms of real vectors as $\ket{v}_{\mu}=\ket{v_R}_{\mu}+i \ket{v_I}_{\mu}$.
Then, the reality condition $v[f]=v[f]^*$ for $f$ real implies
\begin{equation}
0=
\langle v|f\rangle_{\mu}-\langle v|f\rangle_{\mu}^*
=
\left(\langle v_R|f\rangle_{\mu}
    -
    i\langle v_I|f\rangle_{\mu}\right)
    -
\left(
    \langle v_R|f\rangle_{\mu}
    +
    i\langle v_I|f\rangle_{\mu}\right)
    =-2i\langle v_I|f\rangle_{\mu}
    \,,
\end{equation}
for every real $f$, where we have used that the $\mu$ inner product of two real vectors is real.
By complex linearity, we have $\langle v_I|f\rangle_{\mu}=0$ also for complex  $f$, so clearly $\ket{v_I}_{\mu}$ vanishes.
} vector $\ket{v}_\mu$  within the Hilbert space of smearing functions $\K_\mu$;  that is,
\begin{equation}\label{eq:Riesz-VEV}
    v[f]=\braket{v}{f}_\mu.
\end{equation}
Hence, it is natural to ask whether there is a choice of $\psi$ such that we have $\omega_{\psi}(\phi[f]) = -v[f]$, or equivalently,
\begin{equation}
    2\braket{R\psi}{f}_\mu=-\braket{v}{f}_\mu,
\end{equation}
for all $f \in \mathcal{K}_{\mu}$.

For a generic state $\omega$, the image of $R$ need not be dense in ${\mathcal K}_{\mu}$, so the above equation cannot be solved for $\ket\psi_\mu$.
If $\omega$ is pure, however, then $R$ is invertible,\footnote{See \cite[appendix A.8]{Caminiti:paper-2} for a discussion of the properties of $R$ in various cases.} with inverse $-R$, so that the solution is $\ket{\psi}_\mu=R\ket v _\mu/2$.
Therefore, in the pure case, boundedness of $v$ is sufficient to guarantee that $\omega^v$ be excitable from $\omega$, and an explicit representative is given by
\begin{equation}\label{eq:pure-avatar}
    \ket*{\omega^v}=e^{\frac{i}{2}\phi[Rv]}\ket{\omega}.
\end{equation}

\subsection{Abelian excitability}

Let $\omega$ be a zero-mean Gaussian state on an abelian free field algebra, \ie a free $*$-algebra with $\Omega=0$.
In this case, the coherent state ansatz adopted in the pure case clearly cannot work: not only is $R$ not invertible --- in fact it vanishes identically --- but we also cannot hope to excite $\omega^v$ out of $\omega$ by acting with Weyl unitaries, since these lie in the center and hence leave expectation values unchanged.

Motivated by the fact that it was the exponential nature of the ansatz in \eqref{eq:pure-avatar} which encoded the combinatorics needed to reproduce $\omega^v$, we now consider an ansatz of the form  
\begin{equation}\label{eq:abel-ansatz}
    \ket*{\omega^v}=e^{-\braket{\chi}_\mu}e^{\phi[\chi]}\ket\omega
\end{equation}
for some real $\ket{\chi}$ in $\K_\mu$.
Here the c-number prefactor is required to ensure that the vector has unit norm.\footnote{The exponential of the self-adjoint operator $\phi[\chi]$ is defined by the spectral theorem, and is itself a positive self-adjoint operator, albeit generally only densely defined on $\H_\omega$. The GNS vector $\ket{\omega}$ belongs to its domain, as witnessed by the finiteness of the normalization factor; the action of $e^{\phi[\chi]}$ on $|\omega\rangle$ can be computed explicitly as a convergent power series as in appendix \ref{app:abelian-excitability-combinatorics}.} 

In appendix \ref{app:abelian-excitability-combinatorics}, we compute the general identity
\begin{equation}
\label{eq:weyl-in-exp-g}
    \langle e^{\phi[g]}\omega|e^{i \phi[f]}|e^{\phi[g]}\omega\rangle
    =
e^{-\frac{1}{2}\omega(\phi[f-2i g]^2)}
\end{equation}
for any $f$ and $g$ real.
This implies 
\begin{equation}
\begin{aligned}
\label{eq:abelian-ansatz-check}
        \expval*{e^{i\phi[f]}}{\omega^v} 
        &=e^{-2\braket{\chi}_\mu-\frac{1}{2}\omega(\phi[f-2i\chi]^2)}.
\end{aligned}
\end{equation}
Expanding the exponent in terms of the $\mu$ inner product \eqref{eq:mu-def} and using the commutativity of the algebra then gives
\begin{equation}
\begin{aligned}
        \expval*{e^{i\phi[f]}}{\omega^v} &= e^{-\frac{1}{2}\braket{f}_\mu}e^{2i\braket{\chi}{f}_\mu}\\
        &=\expval*{e^{i\phi[f]}}{\omega} e^{2i\braket{\chi}{f}_\mu}.
\end{aligned}
\end{equation}
Comparing to equation \eqref{eq:Gaussian-state}, we see that for the exponential ansatz \eqref{eq:abel-ansatz} to be a representative of $\omega^v$, it remains to tune $\ket{\chi}_\mu$ so that the correct VEV shift is enforced:
\begin{equation}
    2 \langle \chi |f\rangle_{\mu} =  -  \langle v|f\rangle_{\mu}\\,
\end{equation}
where we are again using the discussion above equation \eqref{eq:Riesz-VEV} to express the bounded VEV $v$ in terms of a real vector $\ket{v}_{\mu}$.
Clearly it suffices to pick $\ket{\chi}_\mu = -\ket v_\mu /2$.

To summarize, we find that for abelian free field algebras, boundedness of $v$ is sufficient to guarantee that $\omega^v$ is excitable from $\omega$, and an explicit representative is given by 
\begin{equation}
    \ket*{\omega^v} = e^{-\frac{1}{4}\braket{v}_\mu}e^{-\frac{1}{2}\phi[v]}\ket{\omega}.
\end{equation}

\subsection{Centrally pure excitability}

We now address the case in which $\omega$ is centrally pure by combining the results of the previous two subsections.
We begin by decomposing the completed phase space $\mathcal{K}_{\mu}$ into the kernel of $R$ and its orthogonal complement:
\begin{equation}
    \mathcal{K}_{\mu}=\mathcal{K}_{\mu}^{0}\oplus \mathcal{K}_{\mu}^\perp\,.
    \label{eq:orthogonal-decomp-of-K-mu}
\end{equation}
Correspondingly, we decompose 
the VEV-shift $\ket{v}_\mu$ as
\begin{equation}\label{eq:split-centr-pure}
    \ket v_\mu = \ket{v_0}_\mu + \ket{v_\perp}_\mu.
\end{equation}
We then claim that the desired representative of $\omega^v$ is given by  the following vector in $\H_\omega$:
\begin{equation}\label{eq:avatar-centr-pure}
    \ket{\omega^v}= e^{-\frac{1}{4}\braket{v_0}_\mu}e^{\frac{i}{2}\phi[Rv_\perp]}e^{-\frac{1}{2}\phi[v_0]}\ket{\omega}.
\end{equation}

The rationale behind this choice is as follows. 
Without modifying the von Neumann algebra of physical observables, one can extend the $\ast$-algebra so as to also include field operators $\phi[\psi]$ for every real $\ket{\psi}_{\mu}\in \K_\mu$. 
The resulting $\ast$-algebra factorizes into an abelian algebra generated by the kernel of $R$, and an algebra generated by its orthogonal complement.
Moreover, since $\psi_0$ and $\psi_{\perp}$ are real, and since $\phi[\psi_0]$ and $\phi[\psi_{\perp}]$ commute, mixed correlators of the form $\omega(\phi[\psi_{\perp}]\phi[\psi_{0}])$
reduce to the associated $\mu$ inner products, which vanish by definition of the decomposition \eqref{eq:orthogonal-decomp-of-K-mu}.
Consequently, the state $\omega$ factorizes into its restrictions to these two subalgebras, namely as $\omega \cong \omega_0\otimes \omega_\perp$, with $\omega_\perp$ pure. The shifted state $\omega^v$ factorizes analogously as the product of these restrictions with one-point functions shifted by the two summands in \eqref{eq:split-centr-pure}, respectively; namely, into  $\omega_0^{v_0}\otimes \omega_\perp^{v_\perp}$.
The state $|\omega^v\rangle$ in equation \eqref{eq:avatar-centr-pure} is simply the tensor product of the representatives of the individual states $\omega_0^{v_0}$ and $\omega_{\perp}^{v_{\perp}}.$

To explicitly verify that equation \eqref{eq:avatar-centr-pure} is indeed a representative for $\omega^v$, it suffices to check that it gives the correct expectation values for Weyl operators $e^{i\phi[f]}$, with real smearing functions $f$. 
We begin by writing these expectation values explicitly as
\begin{equation}
    \begin{aligned}
        \expval*{e^{i\phi[f]}}{\omega^v}
        &= e^{- \frac{1}{2} \langle v_0 | v_0\rangle_{\mu}}
        \langle e^{- \frac{1}{2} \phi[v_0]} \omega | e^{-\frac{i}{2} \phi[R v_{\perp}]} e^{i \phi[f]} e^{\frac{i}{2} \phi[R v_{\perp}]} | e^{- \frac{1}{2} \phi[v_0]} \omega \rangle.
    \end{aligned}
\end{equation}
We can combine the unitary operators by repeated application of the Baker-Campbell-Hausdorff formula to obtain
\begin{equation}
    \begin{aligned}
        \expval*{e^{i\phi[f]}}{\omega^v}
        &= e^{- \frac{1}{2} \langle v_0 | v_0\rangle_{\mu}} e^{i \langle f | R^2 | v_{\perp}\rangle_{\mu}}
        \langle e^{- \frac{1}{2} \phi[v_0]} \omega | e^{i \phi[f]}  | e^{- \frac{1}{2} \phi[v_0]} \omega \rangle.
    \end{aligned}
\end{equation}
We then apply our general formula \eqref{eq:weyl-in-exp-g} to compute the remaining expectation value as
\begin{equation}
    \begin{aligned}
        \expval*{e^{i\phi[f]}}{\omega^v}
        &= e^{- \frac{1}{2} \langle v_0 | v_0\rangle_{\mu}} e^{i \langle f | R^2 | v_{\perp}\rangle_{\mu}}
        e^{- \frac{1}{2} \omega(\phi[f+ i v_0]^2)}.
    \end{aligned}
\end{equation}
To simplify, we split $f$ according to decomposition \eqref{eq:orthogonal-decomp-of-K-mu} as $f = f_0 + f_{\perp}$, and use the fact that $\phi[v_0]$ commutes with $\phi[f]$ to rewrite this as
\begin{equation}
    \begin{aligned}
        \expval*{e^{i\phi[f]}}{\omega^v}
        &= e^{- \frac{1}{2} \langle v_0 | v_0\rangle_{\mu}} e^{i \langle f | R^2 | v_{\perp}\rangle_{\mu}}
        e^{-\frac{1}{2} (\langle f | f \rangle_{\mu} - \langle v_0 | v_0 \rangle_{\mu})}e^{- i \langle v_0 | f_0 \rangle_{\mu}}.
    \end{aligned}
\end{equation}
Simplifying, and using that $R$ squares to $-1$ on its support, we have
\begin{equation}
    \begin{aligned}
        \expval*{e^{i\phi[f]}}{\omega^v}
        &= e^{-\frac{1}{2} \langle f | f \rangle_{\mu}} e^{- i \langle f_{\perp} | v_{\perp}\rangle_{\mu}}
        e^{- i \langle v_0 | f_0 \rangle_{\mu}}.
    \end{aligned}
\end{equation}
We may finally use the form of the Weyl-operator expectation values of $\omega,$ together with the fact that the $\mu$ inner product is symmetric on real vectors, to obtain the desired equality
\begin{equation}
    \begin{aligned}
        \expval*{e^{i\phi[f]}}{\omega^v}
        &= \langle \omega | e^{i \phi[f]} |\omega\rangle e^{- i \langle v | f\rangle_{\mu}}.
    \end{aligned}
\end{equation}

Since our final expression coincides with the desired expectation  value in the algebraic state $\omega^v$, we conclude that for centrally pure states, $\mu$-boundedness of the VEV-shift is sufficient to construct an explicit representative for the excitability relation $\omega^v\prec \omega$.

\subsection{General excitability}

To address the problem of exciting a VEV shift $v$ out of a general zero-mean Gaussian state $\omega$, we pass to canonical purifications, thereby reducing the question to the centrally pure VEV-excitability problem solved in the previous subsection.
Here, we are using that by \cite[theorem 1.11]{Caminiti:paper-2}, excitability is preserved under canonical purification: one state can be excited out of another if and only if the same relation holds for their canonical purifications.
Moreover, by \cite[lemma 4.1]{Sorce:paper1}, canonical purifications are always centrally pure.

In the present setting, where $\omega$ has zero mean, the canonical purification of $\omega$  is a centrally pure zero-mean Gaussian state $\hat \omega$.
Meanwhile, the canonical purification of the VEV-shifted state $\omega^v$ is the centrally pure Gaussian state $\hat \omega^{\hat v}$ obtained by shifting the one-point function of  $\hat \omega$ by
\begin{equation}\label{eq:hat-v}
    \hat v[f_1\oplus f_2]\equiv v[f_1]+v[f_2],
\end{equation}
as in equation \eqref{eq:omega-hat-VEV}.
To finish the proof of lemma \ref{lem:2}, we must show that if the original VEV shift $v$ is $\mu$-bounded, then $\hat \omega^{\hat v}$ can be excited out of $\hat \omega$. 
Since this is a centrally pure excitability question, we may employ the result of the preceding subsection --- our aim is therefore to show that if $v$ is $\mu$-bounded, then $\hat{v}$ is $\hat{\mu}$-bounded.

Going forward, we will assume that $v$ is $\mu$-bounded.
To show $\hat{\mu}$-boundedness of $\hat{v}$, we will seek a vector $|\hat{v}\rangle_{\hat{\mu}}$ satisfying
\begin{equation}\label{eq:Riesz-hat}
    \hat v[f_1\oplus f_2] = \braket{\hat v }{f_1\oplus f_2}_{\hat \mu}.
\end{equation}
We look for such a representative in the form
\begin{equation}
    \ket*{\hat v}_{\hat \mu} = \ket{u_1 \oplus u_2}_{\hat \mu}.
    \label{eq:ansatz-for-v-hat}
\end{equation}
Using $\mu$-boundedness of $v$ to produce a vector $|v\rangle_{\mu}$ implementing the VEV $v$ within $\K_{\mu}$, the equation to satisfy reads
\begin{equation}
    \begin{aligned}
        \braket{v}{f_1}_\mu +\braket{v}{f_2}_\mu  &=\braket*{\hat v}{f_1\oplus f_2 }_{\hat \mu}\\
        &= \braket{u_1}{f_1}_\mu +\braket{u_2}{f_2}_\mu + \braket*{u_1}{\sqrt{1+R^2}f_2}_\mu+\braket*{u_2}{\sqrt{1+R^2}f_1}_\mu\\
        &= \braket*{u_1 + \sqrt{1+R^2}u_2}{f_1}_\mu + \braket*{u_2 + \sqrt{1+R^2}u_1}{f_2}_\mu.
    \end{aligned}
\end{equation}
In the second equality we have used the explicit form of the $\hat \mu$ inner product, given in equation \eqref{eq:mu-hat}; in the third equality, we have used the fact that $R$ is anti-hermitian.

It follows that the ansatz in equation \eqref{eq:ansatz-for-v-hat} represents $\hat v$ if we choose  
\begin{equation}
\ket{u_1}_\mu=\ket{u_2}_\mu=\big(1+\sqrt{1+R^2}\big)^{-1}\ket{v}_\mu.
\end{equation}
The inverse in this expression exists as a bounded operator because  the spectrum of $R$ is contained between $-i$ and $+i$.
This gives the required Riesz representative $\ket{\hat v}_{\hat\mu}$, thereby showing that $\hat v$ is $\hat\mu$-bounded, as desired.

\section{Application: free fields in AdS}
\label{sec:large-N}

So far, we have developed a general understanding of excitability relations among Gaussian states with VEVs in (generalized) free field theories.
In this section, we apply our results to the case of free fields in AdS.
We find that for a VEV to be excited from the Poincar\'e AdS vacuum, its Klein-Gordon mode profile must be $L^2$-normalizable in momentum space.
We then show explicitly that this is equivalent to holographic excitability of the boundary VEV defined via the extrapolate dictionary.

\subsection{Coherent excitability in AdS}

Consider a free Klein-Gordon field in the Poincaré patch of $(d+1)$-dimensional anti-de Sitter space, with Poincaré coordinates $X=(z, x)$.
To define the free field $\ast$-algebra $\A_0$, it is not enough to specify the Klein-Gordon equation in the bulk: one must also choose a boundary condition at the conformal boundary $z=0$ that uniquely specifies the \ql advanced-minus-retarded\qr kernel defining the commutator $\Omega$.

Following \cite{BF, Balasubramanian:Lorentzian, Klebanov:1999tb}, we label this choice by a parameter $\nu>-1$ related to the field mass by $m^2=-d^2/4 +\nu^2$, and impose
\begin{equation}
    \lim_{z\to 0} z^{-\Delta}\phi(z,  x) = \text{finite}, \qquad \Delta = d/2 + \nu.
    \label{eq:boundary-condition-ads}
\end{equation}
As noted first in \cite{BF}, we must have $\nu > -1$ for consistency of the theory, but both choices of $\nu$ satisfying the equation $m^2 = -d^2/4 + \nu^2$ are allowed in the ``BF regime'' $-d^2/4 < m^2 < 1 - d^2/4.$
The boundary condition of equation \eqref{eq:boundary-condition-ads}, together with a choice of $\nu$ if one is in the BF regime, fixes the commutator $\Omega$ to be\footnote{Explicit computation shows that if one takes the canonically normalized conjugate momentum $\pi=z^{1-d}(\partial/\partial t)\phi$, and restricts to an equal time slice, then the expression $-i\Omega(X_1, X_2)$ reproduces the expected equal-time canonical commutation relation $[\phi, \pi] = i \delta$.
}
 \begin{equation}
    \Omega(X_1, X_2)= 2\pi(z_1z_2)^{d/2}\int_{V_+^*}\frac{\dd{^d k}}{(2\pi)^d}J_\nu(z_1\sqrt{|k^2|})J_\nu(z_2\sqrt{|k^2|})\sin(k_{\mu}(x_1- x_2)^{\mu})\,.
    \label{eq:ads-Omega}
\end{equation}
In this expression, $V_+^*\equiv (V_+)^*$ is the dual of the future light cone in $d$-dimensional Minkowski spacetime; i.e., we have $k_\mu t^{\mu} \geq 0$ for all future-pointing causal vectors $t^{\mu}.$
Together with the bulk equation of motion, equation \eqref{eq:ads-Omega} defines the Klein-Gordon $\ast$-algebra $\A_0$ through the canonical commutation relations.

The AdS-invariant vacuum is the zero-mean Gaussian state on $\A_0$ with two-point function defined distributionally by\footnote{See for example \cite[equation $(4.6)$]{Duetsch:2002hc}, though note that this reference uses mostly-minuses signature.
}
 \begin{equation}
    \omega (\phi(X_1)\phi(X_2))= \pi (z_1z_2)^{d/2}\int_{V_+^*}\frac{\dd{^d k}}{(2\pi)^d}J_\nu(z_1\sqrt{|k^2|})J_\nu(z_2\sqrt{|k^2|})e^{-i k_{\mu}(x_1- x_2)^{\mu}}\,,
    \label{eq:ads-2pt-func}
\end{equation}
where $J_{\nu}$ is a Bessel function of the first kind.\footnote{To see that this is the correct expression for the vacuum, one first checks that it satisfies the equations of motion for the Klein-Gordon field in AdS$_{d+1},$ with boundary conditions matching those in equation \eqref{eq:boundary-condition-ads}, and appropriate regularity conditions in the bulk.
One then checks that it is invariant under the identity-connected isometries of Poincar\'{e} AdS$_{d+1}$,
which in particular restricts the integration domain to be $V_+^*$, $V_-^*$, or their union.
The choice to integrate only over $V_+^*$ is a form of the spectrum condition appropriate for AdS$_{d+1}$ --- it guarantees that the two-point function is only supported on ``physical particles.''
Finally, the overall normalization is fixed by consistency with the commutator \eqref{eq:ads-Omega} and the state normalization $\omega(1)=1$.}

To smear $\omega$ against bulk test functions $F_1$, $F_2$, we use the volume form $\sqrt{|g|}\dd{^{d+1}X} = z^{-(d+1)}\dd{z}\dd{^dx},$ and write
\begin{align}
    \begin{split}
    \frac{1}{\pi} \omega(\phi[F_1]\phi[F_2])
        & = \int_{V_+^*}\frac{\dd{^d k}}{(2\pi)^d} \int \dd{z_1} \dd{^dx_1} z_1^{-(d+1)}\int \dd{z_2} \dd{^dx_2} z_2^{-(d+1)}(z_1 z_2)^{d/2}\\
        & \qquad \qquad \times J_\nu(z_1\sqrt{|k^2|})J_\nu(z_2\sqrt{|k^2|})e^{-i k_{\mu}(x_1- x_2)^{\mu}} F_1(z_1,x_1) F_2(z_2,x_2)\,.
    \end{split}
    \label{eq:omega-AdS-expanded}
\end{align}
Evaluating the Fourier transforms in the $x$ coordinates, this reads
\begin{align}
    \begin{split}
    \frac{1}{\pi} \omega(\phi[F_1]\phi[F_2])
        & = \int_{V_+^*}\frac{\dd{^d k}}{(2\pi)^d} \int \dd{z_1} \dd{z_2}(z_1z_2)^{d/2-(d+1)} \\
        & \qquad \qquad \times J_\nu(z_1\sqrt{|k^2|})J_\nu(z_2\sqrt{|k^2|}) \hat{F}_1(z_1, k) \hat{F}_2(z_2, -k).
    \end{split}
\end{align}
We can then simplify this expression in terms of the Hankel-type transform,
\begin{equation}
\begin{aligned}
    \mathcal{F}(k)&=\int_0^{\infty} \dd{z}  z^{-d/2-1}
    J_\nu(z\sqrt{|k^2|})
    \hat{F}(z, k)\,,
\end{aligned}
\end{equation}
as
\begin{align}
    \begin{split}
    \omega(\phi[F_1]\phi[F_2])
        & = \pi \int_{V_+^*}\frac{\dd{^d k}}{(2\pi)^d} \mathcal{F}_1(k) \mathcal{F}_2(-k).
    \end{split}
    \label{eq:omega-AdS-collapsed}
\end{align}
Note that if $F$ is real, then we have $\hat{F}(z,k)^*=\hat{F}(z,-k)$ and $\mathcal{F}(k)^*=\mathcal{F}(-k)$.
Accordingly, the $\mu$ norm in the bulk, for a real smearing function $F=F(z, x)$, is proportional to the $L^2(V_+^*)$ norm of $\mathcal{F}=\mathcal F(k)$:
\begin{equation}
    \braket{F}_{\mu}
= \pi \int_{V_+^*}\frac{\dd{^dk}}{(2\pi)^d}|\mathcal{F}(k)|^2
\,.
\label{eq:mu-norm-ads-hankel}
\end{equation}

Consider now a VEV shift of the AdS vacuum $\omega$ by a classical profile $v_s(X)$.
The connected two-point function of $\omega$ is unchanged, and remains given by the right-hand side of equation \eqref{eq:ads-2pt-func}.
Note that $v_s(X)$ must be a real solution of the wave equation satisfying the chosen boundary conditions, so it is fully characterized by a momentum-space mode $\tilde{v}_s(k)$ via the formula
\begin{equation}\label{eq:momentum-space-VEV}
    v_s(X) = \int_{V_+^*} \frac{\dd^dk}{2(2\pi)^d}
    z^{d/2}J_{\nu}(z\sqrt{|k^2|})\left[\tilde{v}_s(k)e^{i k_{\mu} x^{\mu}}
    +
    \tilde{v}_s(-k)e^{-i k_{\mu} x^{\mu}}\right]\,,
\end{equation}
with $\tilde{v}_s(-k) = \tilde{v}_s(k)^*$.
In turn, this expression characterizes the VEV shift functional $v[F]$ via
\begin{equation}\label{eq:VEV-bulk-AdS}
\begin{aligned}
    v[F]&=\int \dd{^{d+1}X}\sqrt{|g|}v_s(X)F(X)\\
    &=\int_{V_+^*} \frac{\dd^dk}{2(2\pi)^d}
 \left[\tilde{v}_s(-k) \mathcal{F}(k)+ \tilde{v}_s(k) \mathcal{F}(-k)\right]\\
    &=\text{Re}\int_{V_+^*} \frac{\dd^dk}{(2\pi)^d} \tilde{v}_s(k)\mathcal{F}(-k)\,,
\end{aligned}
\end{equation}
where in the last step we have assumed that $F$ is real.

Our goal is to show that the VEV $v_s$ can be excited out of the vacuum if and only if one has
\begin{equation}
    \frac{1}{\pi} \int_{V_+^*} \frac{\dd^dk}{(2\pi)^d}|\tilde{v}_s(k)|^2< \infty\,.
    \label{eq:excitability-condition-ads}
\end{equation}
As we saw in the preceding sections, VEV excitability is equivalent to the requirement that $v[F]$ defines a bounded functional with respect to the $\mu$ inner product. Since $v$ is complex-linear,
it is enough to check boundedness on real smearing functions with respect to the real inner product $\Re\braket*{\,\cdot\,}{\,\cdot\,}_\mu$.
In what follows, all test functions are accordingly taken to be real unless otherwise stated.

First, assume that equation \eqref{eq:excitability-condition-ads} holds.
Then the last line of equation \eqref{eq:VEV-bulk-AdS} may be viewed as an $L^2(V_+^*)$ overlap, and applying the Cauchy-Schwarz inequality gives
\begin{equation}
\begin{aligned}
    |v[F]|^2 
    &\le  
    \int_{V_+^*} \frac{\dd^dk}{(2\pi)^d}|\tilde{v}_s(k)|^2 
    \int_{V_+^*} \frac{\dd^dp}{(2\pi)^d}|\mathcal{F}(p)|^2 \\
    &= \bigg( \frac{1}{\pi}
    \int_{V_+^*} \frac{\dd^dk}{(2\pi)^d}|\tilde{v}_s(k)|^2 \bigg) \braket{F}_\mu.
\end{aligned}
\end{equation}
Thus \eqref{eq:excitability-condition-ads} is sufficient for $v$ to be $\mu$-bounded, and hence for the VEV shift to be excitable.

Conversely, suppose that the VEV shift \eqref{eq:VEV-bulk-AdS} can be excited out of the AdS vacuum.
Then $v$ is bounded on the space of real smearing functions with respect to $\Re\braket{\,\cdot\,}{\,\cdot\,}_\mu$.
But by equation \eqref{eq:omega-AdS-collapsed}, the $\mu$ inner product on real smearing functions is simply
\begin{equation}
    \langle F_1 | F_2 \rangle_{\mu}
        = \frac{1}{2} \left[ \omega(\phi[F_1] \phi[F_2]) + \omega(\phi[F_2] \phi[F_1]) \right]
        = \pi\, \text{Re} \int_{V_+^*}\frac{\dd{^d k}}{(2\pi)^d} \mathcal{F}_1(k)^* \mathcal{F}_2(k).
\end{equation}
So the real $\mu$-Hilbert space is unitarily equivalent to (a rescaling of) the real Hilbert space version of $L^2(V_+^*).$
Since $v$ is bounded as a functional on this space, the Riesz representation theorem tells us that there must exist some $L^2$ function $\mathcal{V}(k)$ satisfying
\begin{equation}
    v[F]
        = \pi\, \text{Re} \int_{V_+^*} \frac{\dd{^d k}}{(2\pi)^d} \mathcal{V}(k) \mathcal{F}(-k).
\end{equation}
Comparing to equation \eqref{eq:VEV-bulk-AdS} gives $\tilde{v}_s(k) = \pi \mathcal{V}(k),$ so $L^2$ normalizability of $\mathcal{V}(k)$ implies $L^2$ normalizability of $\tilde{v}_s(k).$

\subsection{Boundary dual perspective}

We now demonstrate that the VEV excitability condition \eqref{eq:excitability-condition-ads} for free fields in AdS has a simple boundary interpretation through the holographic dictionary.

In the semiclassical regime, where backreaction is suppressed,
the holographic dictionary implies a duality between free fields $\phi(X)$  in AdS  and generalized free fields $\varphi(x)$ living on the conformal boundary --- see e.g.~\cite{Duetsch:2002hc} for an explanation of this duality in algebraic terms. 
The key formula underlying the duality is the extrapolate dictionary, namely the operator equation
\begin{equation}
    \varphi(x) = N_\nu  \lim_{z\to 0}z^{-\Delta} \phi(z,x)\,,
    \label{eq:extrapolate}
\end{equation}
where we have added a factor of $N_\nu=2^\nu \Gamma(\nu+1)$ as a normalization convention. 
In the case of the Poincaré patch, the boundary is $d$-dimensional Minkowski spacetime,
and the boundary $\ast$-algebra $\A_{\partial, 0}$ is generated by fields $\varphi[f]$ obeying 
\begin{equation}
    [\varphi[f], \varphi[g]]=-i\Omega_\partial[f, g],
\end{equation}
where, following equation \eqref{eq:extrapolate}, $\Omega_\partial$ is defined distributionally by 
\begin{equation}
    \Omega_\partial(x_1, x_2) = N_\nu^{2}\lim_{z_j \to 0}\left[z_1^{-\Delta}z_2^{-\Delta}\Omega(X_1, X_2)\right],
\end{equation}
and where $\Omega$ is the \ql advanced-minus-retarded'' kernel \eqref{eq:ads-Omega} in the Poincaré patch.
Note that $\Omega_\partial$ is not  induced by any Klein-Gordon equation in the boundary Minkowski spacetime, which means that the boundary algebra is that of a generalized free field.

The  vacuum state $\omega$ in the bulk defines a natural zero-mean Gaussian state $\omega_\partial$ on $\A_{\partial, 0}$ via the extrapolate dictionary: 
\begin{equation}
\begin{aligned}
        \omega_\partial(\varphi(x_1)\varphi(x_2))  &=  N^{2}_\nu\lim_{z_j\to 0}\left[z_1^{-\Delta}z_2^{-\Delta}\omega(\phi(X_1)\phi(X_2))\right]\\
        &=\pi\int_{V_+^*}\frac{\dd{^d k}}{(2\pi)^d} |k^2|^{\nu}e^{-ik_\mu(x_1-x_2)^\mu}\,.
\end{aligned}
\end{equation}
By checking the behavior under scalings $x \mapsto \lambda x$, one sees that this is the appropriate two-point function for a conformal primary of dimension $\Delta=d/2 + \nu$.
Smearing with boundary test functions gives 
\begin{equation}
\omega_\partial(\varphi[f_1]\varphi[f_2])= \pi\int_{V_+^*}\frac{\dd{^d k}}{(2\pi)^d} |k^2|^{\nu}\hat f_1(k)\hat f_2(-k),
 \end{equation}
and it follows that the $\mu_\partial$-norm of a real boundary smearing function $f$ is given by 
\begin{equation} \label{eq:bdry-mu}
    \braket{f}_{\mu_\partial} = \pi \int_{V_+^*}\frac{\dd{^d k}}{(2\pi)^d} |k^2|^{\nu}|\hat f(k)|^2\,.
\end{equation}

We now consider the coherent excitation of the bulk AdS vacuum $\omega$ by the VEV shift $v_s(X)$.
By the extrapolate dictionary, this makes the boundary state acquire a one-point function
\begin{equation}
    \omega_\partial(\varphi(x))=N_{\nu}\lim_{z\to0}z^{-\Delta}\omega(\phi(z,x))\,,
\end{equation}
so that the corresponding boundary VEV shift is given by a limit of equation \eqref{eq:momentum-space-VEV} as
\begin{equation}
\begin{aligned}
    v_{s,\partial}(x)&=N_{\nu}\lim_{z\to 0} z^{-\Delta}v_s(z,  x)\\
                    &= \int_{V_+^*} \frac{\dd^dk}{2(2\pi)^d}
    |k^2|^{\nu/2}\left[\tilde{v}_s(k)e^{ik_\mu x^\mu}
    +
    \tilde{v}_s(-k)e^{-ik_\mu x^\mu}\right]\,.
\end{aligned} 
\end{equation}
Smearing leads to a boundary VEV functional that for real $f$ reads 
\begin{equation}
    v_\partial[f]
    =
    \text{Re} \int_{V_+^*} \frac{\dd^dk}{(2\pi)^d}
    \tilde{v}_s(k)|k^2|^{\nu/2}\hat{f}(-k)\,.
    \label{eq:vev-form-cft}
\end{equation}
We now show that the condition \eqref{eq:excitability-condition-ads}, which characterizes bulk VEV excitability, is also necessary and sufficient for excitability of the boundary VEV \eqref{eq:vev-form-cft}. 
In parallel with the bulk case, it suffices to regard  \eqref{eq:vev-form-cft} as a real functional, and to test its boundedness on real smearing functions with respect to the real inner product $\Re\braket{\,\cdot\,}{\,\cdot\,}_{\mu_\partial}$.
Sufficiency then follows directly from the Cauchy-Schwarz inequality in $L^2(V_+^*)$,
\begin{equation}
\begin{aligned}
    |v_\partial[f]|^2
    &\le
    \int_{V_+^*} \frac{\dd^dk}{(2\pi)^d} |\tilde{v}_s(k)|^2\int_{V_+^*} \frac{\dd^dp}{(2\pi)^d}
    |p^2|^{\nu}|\hat{f}(p)|^2\\
    &=
    \left(\frac{1}{\pi} \int_{V_+^*} \frac{\dd^dk}{(2\pi)^d} |\tilde{v}_s(k)|^2\right)
    \braket{f}_{\mu_\partial}
    \,.
\end{aligned}
\end{equation}

Conversely, suppose that the VEV shift \eqref{eq:vev-form-cft}  is bounded on the real Hilbert space obtained by completing real smearing functions with respect to  $\Re\braket{\,\cdot\,}{\,\cdot\,}_{\mu_\partial}$.
By equation \eqref{eq:bdry-mu},  this Hilbert space can be identified with (a rescaling of) the real version of $L^2(V_+^*)$ modified by the weight $|k^2|^{\nu}$ in momentum space.
The Riesz representation theorem then tells us that there exists some momentum-space function $\hat{v}_{\del}(k)$ satisfying
\begin{equation} \label{eq:v-del-Riesz}
    v_{\del}[f] = \pi\, \text{Re} \int_{V_+^*} \frac{\dd^d k}{(2 \pi)^d} |k^2|^{\nu} \hat{v}_{\del}(k) \hat{f}(-k)
\end{equation}
and
\begin{equation} \label{eq:v-del-normalized}
    \pi \int_{V_+^*} \frac{\dd^d k}{(2 \pi)^d} |k^2|^{\nu} |\hat{v}_{\del}(k)|^2 < \infty.
\end{equation}
Comparing equation \eqref{eq:v-del-Riesz} to equation \eqref{eq:vev-form-cft} gives $\tilde{v}_s(k) = \pi |k^2|^{\nu/2} \hat{v}_{\partial}(k),$ and inequality \eqref{eq:v-del-normalized} gives $L^2$ normalizability of $\tilde{v}_s(k).$

To summarize, we have shown that the condition \eqref{eq:excitability-condition-ads} for exciting a VEV from the Poincar\'e AdS vacuum can be recast as a VEV excitability condition in the dual generalized free theory.
In particular, this proves that whether a bulk VEV is excitable depends only on its asymptotic behavior near the AdS boundary, as witnessed by the extrapolate dictionary \eqref{eq:extrapolate}.
This is compatible with the general principle that failures of excitability in physical states of a free field theory always result from infrared effects, rather than from details of the state at finite distance.\footnote{For instance, see \cite{Verch:hadamard}, where it was shown that excitability is always possible within compact subregions for zero-mean Gaussian states obeying the Hadamard property. } 

\section*{Acknowledgments}
JC acknowledges the support of the Natural Sciences and Engineering Research Council of Canada through Vanier Canada Graduate Scholarships [Funding Reference Number: CGV–192707]. 

 \appendix
\section{A matrix element computation}
\label{app:abelian-excitability-combinatorics}

In this appendix, we derive the identity
\begin{equation}
\label{eq:weyl-in-exp-g-app}
    \langle e^{\phi[g]}\omega|e^{i \phi[f]}|e^{\phi[g]}\omega\rangle
    =
e^{-\frac{1}{2}\omega(\phi[f-2i g]^2)}
\end{equation}
for real smearing functions $f$ and $g$.
This appeared as equation \eqref{eq:weyl-in-exp-g} in the main text.

First note that the power series defining $e^{\phi[g]} |\omega\rangle$ converges in $\mathcal{H}_{\omega}$.\footnote{It suffices to see that $\sum_{n=0}^\infty \norm{\phi[g]^n \ket{\omega}/n!}$ converges, which can be checked using the Gaussianity of $\omega$.}
Using this power series,
together with the commutator
\begin{equation}
    [a[f], \phi[g]^n]= n\,\omega(\phi[f]\phi[g]) \,\phi[g]^{n-1}\,,
\end{equation}
one can also compute the diagonal action of $a[f]$ on $e^{\phi[g]} |\omega\rangle$:
\begin{equation}
a[f] e^{\phi[g]}|\omega\rangle = \omega(\phi[f]\phi[g]) e^{\phi[g]} |\omega\rangle\,.
\label{eq:annihilate-exp-g}
\end{equation}

We now expand the left-hand side of
equation \eqref{eq:weyl-in-exp-g-app} as a power series in $\phi[f]$.
So long as the power series converges absolutely, it gives a rigorous answer for the expectation value.
Using the normal-ordering identity
\begin{equation}
    \phi[f]^n
    =
    \sum_{k=0}^{ \left \lfloor n/2 \right \rfloor} \frac{n!}{k!(n-2k)!}\left(\frac{\braket{f}_\mu}{2}\right)^k :\phi[f]^{n-2k}:
\end{equation}
we have\footnote{
To obtain the right-hand side of equation \eqref{eq:normal-order-exponential}, we exchanged the order of summation over $n$ and $k$, and replaced $n$ by a new summation index $m=n-2k$.
The order of summation can be switched because of the absolute convergence in equation \eqref{eq:processing-matrix-elt} below.
}
\begin{equation}
\label{eq:normal-order-exponential}
    e^{i \phi[f]}
    =
    \sum_{n=0}^\infty \frac{i^n}{n!}\phi[f]^n
    =
e^{-\frac{1}{2}\braket{f}_\mu} \sum_{m=0}^\infty \frac{i^m}{m!} :\phi[f]^m:\,,
\end{equation}
and so 
\begin{equation}
  \langle e^{\phi[g]}\omega|e^{i \phi[f]}|e^{\phi[g]}\omega\rangle
  =
  e^{-\frac{1}{2}\braket{f}_\mu}
  \sum_{m=0}^\infty \frac{i^m}{m!}\langle e^{\phi[g]}\omega| :\phi[f]^m:|e^{\phi[g]}\omega\rangle\,.
\end{equation}
Expanding $:\phi[f]^m:$ via the formula
\begin{equation}
:\phi[f]^m: = \sum_{k=0}^m \binom{m}{k}(a[f]^\dag)^k a[f]^{m-k}\,,
\end{equation}
we can act with the $a[f]^\dag$ on the bra and the $a[f]$ on the ket as in equation \eqref{eq:annihilate-exp-g}.
We obtain 
\begin{equation}
\begin{aligned}
    \langle e^{\phi[g]}\omega|e^{i \phi[f]}|e^{\phi[g]}\omega\rangle
    &=
    e^{-\frac{1}{2}\braket{f}_\mu} \sum_{m=0}^\infty \frac{i^m}{m!} \sum_{k=0}^{m}\binom{m}{k}\langle e^{\phi[g]} \omega |
    \omega( \phi[g]\phi[f])^{k}
    \omega( \phi[f]\phi[g])^{m-k} | e^{\phi[g]} \omega\rangle\\
    &=
    e^{-\frac{1}{2}\braket{f}_\mu} e^{2\braket{g}_\mu} 
    \sum_{m=0}^\infty \frac{i^m}{m!} \sum_{k=0}^{m} \binom{m}{k} 
    \omega( \phi[g]\phi[f])^{k}
    \omega( \phi[f]\phi[g])^{m-k}\\
    &=
    e^{-\frac{1}{2}\braket{f}_\mu} e^{2\braket{g}_\mu} 
\sum_{k=0}^{\infty} \sum_{j=0}^{\infty} \frac{i^{j+k}}{j!k!} 
\omega( \phi[f]\phi[g])^{j}
\omega( \phi[g]\phi[f])^{k}\,,
\end{aligned}
\label{eq:processing-matrix-elt}
\end{equation}
where in the first line, we used that $\omega( \phi[g]\phi[f])$ is the complex conjugate of $\omega( \phi[f]\phi[g])$; in the second line we used that the norm-squared of $e^{\phi[g]} |\omega\rangle$ is $e^{2\braket{g}_\mu}$, and in the third line, we exchanged the order of summation so that $k$ ranges from $0$ to $\infty$ while $m$ ranges from $k$ to $\infty$, 
and introduced $j = m - k$.

The last sum in equation \eqref{eq:processing-matrix-elt} converges absolutely, giving 
\begin{equation}
    \langle e^{\phi[g]}\omega|e^{i \phi[f]}|e^{\phi[g]}\omega\rangle
    =
    e^{-\frac{1}{2}\braket{f}_\mu} e^{2\braket{g}_\mu} 
    e^{i \omega(\phi[f]\phi[g])} 
    e^{i \omega(\phi[g]\phi[f])} 
    \,,
\end{equation}
which is the right-hand side of equation \eqref{eq:weyl-in-exp-g-app}.

\bibliographystyle{JHEP}
\bibliography{bibliography}

\end{document}